\documentclass[mathematics,article,accept,moreauthors,pdflatex,12pt,a4paper]{mdpi}

\lastpage{x}
\doinum{10.3390/------}
\pubvolume{xx}
\pubyear{2015}
\history{Received: xx / Accepted: xx / Published: xx}
\pdfoutput=1

\usepackage{graphicx}
\usepackage{amsfonts}

\newcounter{thm}
\newcounter{ex}
\newcounter{re}
\newtheorem{theorem}[thm]{Theorem}
\newtheorem{lemma}[thm]{Lemma}

\newtheorem{proposition}[thm]{Proposition}
\theoremstyle{mdpidefinition}

\newtheorem{definition}[thm]{Definition}
\newtheorem{notation}[thm]{Notation}

\Title{The Segal--Bargmann Transform for Odd-Dimensional Hyperbolic Spaces}

\Author{Brian C. Hall $^{1, \dagger,}$* and Jeffrey J. Mitchell $^{2,\dagger}$}

\address{%
$^{1}$ University of Notre Dame, Department of Mathematics, 255 Hurley Hall, Notre Dame, IN 46556, U.S.A.\\
$^{2}$ Robert Morris University, Department of Mathematics, 6001 University Boulevard, Moon Township, PA 15108, U.S.A.}

\contributed{$^\dagger$ These authors contributed equally to this work.}

\corres{bhall@nd.edu}

\abstract{We develop isometry and inversion formulas for the Segal--Bargmann transform
on odd-dimensional hyperbolic spaces that are as parallel as possible to the
dual case of odd-dimensional spheres.}


\begin{document}


\section{Introduction}

The Segal--Bargmann transform for Euclidean space was developed in the 1960's
as a unitary map from $L^{2}(\mathbb{R}^{n})$ to an $L^{2}$ space of
holomorphic functions on $\mathbb{C}^{n}$ with respect to a Gaussian measure
\cite{Ba,Segal1,Segal2}. Motivated by work of L. Gross \cite{Gross1}, the
first author introduced an analog of the Segal--Bargmann transform for compact
Lie groups and proved isometry and inversion formulas for it
\cite{Hall1,Hall2,Bull}. The transform is connected to the Segal--Bargmann
transform on an infinite-dimensional Euclidean space \cite{GM} and can be
extended to the group of paths with values in a compact group \cite{HS}. The
transform also arises in the study of two-dimensional Yang--Mills theory on a
spacetime cylinder \cite{Wr,DH,HallYM} and in geometric quantization
\cite{Geoquant} and has been used in the study of quantum gravity (see
\cite{Ashtekar} and \cite{TW} among many others).

The results of \cite{Hall1,Hall2} were extended to the case of a compact
symmetric space by M. Stenzel in \cite{St}. Let $U$ be a simply connected
compact group, let $K$ be the fixed-point subgroup of an involution, and let
$X=U/K.$ (Every simply connected compact group can be thought of in this way;
other examples include spheres and projective spaces.) If $X_{\mathbb{C}%
}:=U_{\mathbb{C}}/K_{\mathbb{C}}$ is the \textquotedblleft
complexification\textquotedblright\ of $X,$ there is a diffeomorphism
$\Phi:T(X)\rightarrow X_{\mathbb{C}}$ given by%
\[
\Phi(x,Y)=\exp_{x}(iY),
\]
where the right-hand side of the formula refers to the analytic continuation
of the geometric exponential map. The image of a single fiber in $T(X)$ in
$X_{\mathbb{C}}$ may be identified with the dual noncompact symmetric space to
$X.$ If, for example, $X$ is an $n$-sphere, the image of each fiber can be
identified with $n$-dimensional hyperbolic space.

Now, for each $t>0,$ we define the \textbf{Segal--Bargmann transform}
$C_{t}:L^{2}(X)\rightarrow\mathcal{H}(X_{\mathbb{C}})$ by%
\[
C_{t}(f)=(e^{t\Delta/2}f)_{\mathbb{C}},
\]
where $e^{t\Delta/2}$ is the (forward) heat operator on $X$ and $(\cdot
)_{\mathbb{C}}$ denotes analytic continuation in the space variable from $X$
to $X_{\mathbb{C}}.$ If $f\in L^{2}(X)$ and $F=C_{t}(f)\in\mathcal{H}%
(X_{\mathbb{C}}),$ the isometry and inversion formulas of \cite{St} take the
following form:%
\begin{align}
\left\Vert f\right\Vert ^{2}  &  =\int_{x\in X}\int_{Y\in T_{x}(X)}\left\vert
F(\exp_{x}(iY/2))\right\vert ^{2}\nu_{2t}(Y)J(Y)~dY~dx\label{compactIsom}\\
f(x)  &  =\int_{Y\in T_{x}(X)}F(\exp_{x}(iY))\nu_{t}(Y)J(Y)~dY.
\label{compactInv}%
\end{align}
Here, $\nu_{t}$ and $J$ are the heat kernel and the Jacobian of the
exponential, respectively, for the dual noncompact symmetric space to $X.$
Note that in (\ref{compactIsom}), we have $\exp_{x}(iY/2)$ and $\nu_{2t}(Y),$
whereas in (\ref{compactInv}) we have $\exp_{x}(iY)$ and $\nu_{t}(Y).$

It is natural to attempt to extend the isometry and inversion formulas to the
case where $X$ is a noncompact symmetric space. In light of the duality
between compact and noncompact symmetric spaces, one would expect---roughly
speaking---to reverse the roles of compact and noncompact in
(\ref{compactIsom}) and (\ref{compactInv}). In attempting to do this, however,
substantial complications quickly arise. The polar decomposition, for example,
is no longer a diffeomorphism and functions of the form $e^{t\Delta/2}f$ do
not extend to all of $X_{\mathbb{C}}$ but only to the Akhiezer--Gindikin
\textquotedblleft crown domain\textquotedblright\ \cite{AG}. (See
\cite{KS1,KS2}.)

One way to work around these difficulties was developed by B. Kr\"{o}tz, G.
\'{O}lafsson, and R. Stanton in \cite{KOS}. Let $G$ be a connected semisimple
Lie group with finite center and let $K$ be a maximal compact subgroup of $G,$
so that $G/K$ (with a $G$-invariant metric) is a Riemannian symmetric space of
the noncompact type. If $f$ is in $L^{2}(G/K)$ and $F=e^{t\Delta/2}f,$ one
defines the \textbf{orbital integral}, given by%
\begin{equation}
\mathcal{O}_{\left\vert F\right\vert ^{2}}(iY)=\int_{G}\left\vert F(g\cdot
\exp_{x_{0}}(iY/2))\right\vert ^{2}~dg, \label{ofDef}%
\end{equation}
for $Y$ in the tangent space $\mathfrak{a}$ to a maximal flat through the
basepoint $x_{0}.$ This function is initially defined only for $Y$ in a
certain bounded domain $2\Omega,$ and it blows up on the boundary of $2\Omega
$. Nevertheless, there is a pseudodifferential \textquotedblleft shift
operator\textquotedblright\ $D$ on $\mathfrak{a}$ that can be used to
eliminate the singularities in $\mathcal{O}_{\left\vert F\right\vert ^{2}}.$
The operator $D$ is defined, essentially, by requiring that it map the
spherical functions for $G/K$ into their Euclidean counterparts, which are
just linear combinations of exponentials. (Compare Eqn. (3.11) to Eqn. (3.12)
in \cite{KOS}.) We are then interested in the operator $\tilde{D},$ which
describes the action of $D$ on the imaginary axis; that is, $\tilde{D}$
satisfies%
\[
(Dg)(iY)=\tilde{D}(g(iY)).
\]
It is shown in \cite{KOS}, using the Gutzmer-type formula of Faraut \cite{Fa},
that $\tilde{D}\mathcal{O}_{\left\vert F\right\vert ^{2}}$ extends without
singularities to all of $\mathfrak{a}.$

The isometry formula of \cite[Thm. 3.3]{KOS} then reads%
\begin{equation}
\left\Vert f\right\Vert ^{2}=\int_{\mathfrak{a}}\tilde{D}\left[
\mathcal{O}_{\left\vert F\right\vert ^{2}}(iY)\right]  w_{t}%
(Y)~dY,\label{kosIsom}%
\end{equation}
where%
\[
w_{t}(Y)=e^{t\left\vert \delta\right\vert ^{2}}\frac{e^{-\left\vert
Y\right\vert ^{2}/(4t)}}{(4\pi t)^{k/2}}.
\]
Here, $k=\dim\mathfrak{a}$ and $\delta$ is half the sum of the positive roots
with multiplicity. We adjust the formula in \cite{KOS} to fit our
normalization of the heat equation and to correct for a minor inconsistency in
\cite{KOS} in how the orbital integral is defined. We also use a different
normalization of the shift operator, which means that we do not need a factor
of the order of the Weyl group in the definition of $w_{t},$ as in \cite{KOS}.
See also \cite{OS1,OS2} for a different approach to the isometry formula on
$G/K$, but which also involves a sort of shift operator.

On the one hand, the isometry formula in (\ref{kosIsom}) is easy to state and
holds for all symmetric spaces of the noncompact type. On the other hand, it
is not parallel to the compact case and does not explicitly involve the
geometry of the dual compact symmetric space. A different approach was
developed by the authors of the current paper \cite{complex1,complex2} in the
case of a noncompact Riemannian symmetric space of the \textquotedblleft
complex type,\textquotedblright\ that is, a space of the form $X=G/K$ where
$G$ is complex semisimple and $K$ is a maximal compact subgroup. (See also
\cite{compactq} for analogous results on compact quotients of such symmetric
spaces.) The results of \cite{complex1,complex2} are extremely parallel to
(\ref{compactIsom}) and (\ref{compactInv}), with $\nu_{t}$ now representing an
\textquotedblleft unwrapped\textquotedblright\ version of the heat kernel on
the dual compact symmetric space to $X,$ except that there is a subtle
cancellation of singularities that allows the formulas to make sense. When one
moves away from the complex case, the singularities become more complicated
and the results of \cite{complex1,complex2} do not extend as stated.

In the present paper, we consider the case in which $X$ is an odd-dimensional
hyperbolic space, $H^{2n+1},$ in which case $\tilde{D}$ is a differential
operator. As our first main goal, we develop an isometry formula that is as
parallel as possible to the compact case. We do this by taking the adjoint
(with respect to certain natural inner products) of the operator $\tilde{D}$
in (\ref{kosIsom}), which is done by integrating by parts. Two main
observations allow us to construct an isometry formula similar to
(\ref{compactIsom}).

\begin{itemize}
\item If $\tilde{D}^{\ast}$ is the adjoint of $\tilde{D}$, then $\tilde
{D}^{\ast}(w_{t})$ is an \textquotedblleft unwrapped\textquotedblright%
\ version of the heat kernel for the dual compact symmetric space $S^{2n+1}.$

\item Taking the adjoint of $\tilde{D}$ involves boundary terms, which are
meromorphic functions of the radius. These functions tend to zero as long as
we avoid poles at integer multiples of $\pi.$
\end{itemize}

As our second main goal, we develop an inversion formula. We begin by
developing a general inversion formula, in the style of \cite{KOS}, that
involves a shift operator and applies to an arbitrary symmetric space of the
noncompact type. We then specialize to the case of odd-dimensional hyperbolic
spaces, integrate by parts, and obtain an inversion formula similar to
(\ref{compactInv}).

\section{Main Results}

We let $G/K=H^{2n+1}$ denote the hyperbolic space of dimension $2n+1,$ with
the metric normalized to have constant sectional curvature $-1,$ where $G$ is
the identity component of $SO(2n+1,1)$ and $K=SO(2n+1).$ We refer to Section
5.7 of \cite{Da} for standard formulas involving the metric and the Laplacian
on $H^{2n+1}$. We consider the map from $T(H^{2n+1})\rightarrow G_{\mathbb{C}%
}/K_{\mathbb{C}}$ given by%
\begin{equation}
(x,Y)\mapsto\exp_{x}(iY),\label{expMap}%
\end{equation}
where the right-hand side of the above formula refers to the analytic
continuation of the geometric exponential map. For all sufficiently small $R,$
the map (\ref{expMap}) is a diffeomorphism of the set $\left\{
(x,Y)|~\left\vert Y\right\vert <R\right\}  $ onto its image. We refer to the
set of points of the form $\exp_{x}(iY)$ with $\left\vert Y\right\vert <R$ as
a \textbf{tube} in $G_{\mathbb{C}}/K_{\mathbb{C}}.$

If we analytically continue the metric tensor from $G/K=H^{2n+1}$ to a tube in
$G_{\mathbb{C}}/K_{\mathbb{C}}$ and then restrict to the image of a single
fiber in $T(G/K),$ the result is the negative of a Riemannian metric
\cite[Prop. 1.17]{LGS}. Under this metric, the fibers are locally isometric to
the unit sphere $S^{2n+1}.$ We begin by introducing the relevant density to be
used in the fibers.

\begin{definition}
\label{nut.def}For each nonnegative integer $n$ and $t,r\in\mathbb{R}$ with
$t>0,$ let $\nu_{t}(r)$ be given by%
\begin{equation}
\nu_{t}(r)=e^{tn^{2}/2}\left(  -\frac{1}{2\pi}\frac{1}{\sin r}\frac{d}%
{dr}\right)  ^{n}\frac{e^{-r^{2}/(2t)}}{\sqrt{2\pi t}}. \label{nutDef}%
\end{equation}

\end{definition}

We refer to this function as the \textbf{unwrapped heat kernel} for
$S^{2n+1}.$ The $2\pi$-periodization of this function with respect to $r$ is
the actual heat kernel on $S^{2n+1}$ (compare Proposition
\ref{heatKernels.prop}). The unwrapped heat kernel is nonsingular at the
origin but has a pole of order $2n-1$ at nonzero integer multiples of $\pi.$
Our results are based on the idea of taking a limit that \textquotedblleft
stays a fixed distance away from the poles.\textquotedblright\ 

\begin{notation}
\label{Sepsilon.notation}A limit as $R\rightarrow+\infty$ staying a fixed
distance from the poles means a limit as $\operatorname{Re}R$ tends to
$+\infty$ in a region of the form%
\[
S_{\varepsilon,A}=\left\{  \left.  R\in\mathbb{C}\right\vert \operatorname{Re}%
R>0,~\left\vert \operatorname{Im}R\right\vert <A,~\left\vert R-n\pi\right\vert
>\varepsilon,\text{~}n=1,2,3,\ldots\right\}  ,
\]
where $0<\varepsilon<A<\pi.$ See Figure \ref{sepsilon.fig}.
\end{notation}

%

\begin{figure}[ptb]%
\centering
\includegraphics[
height=0.9954in,
width=3.659in
]%
{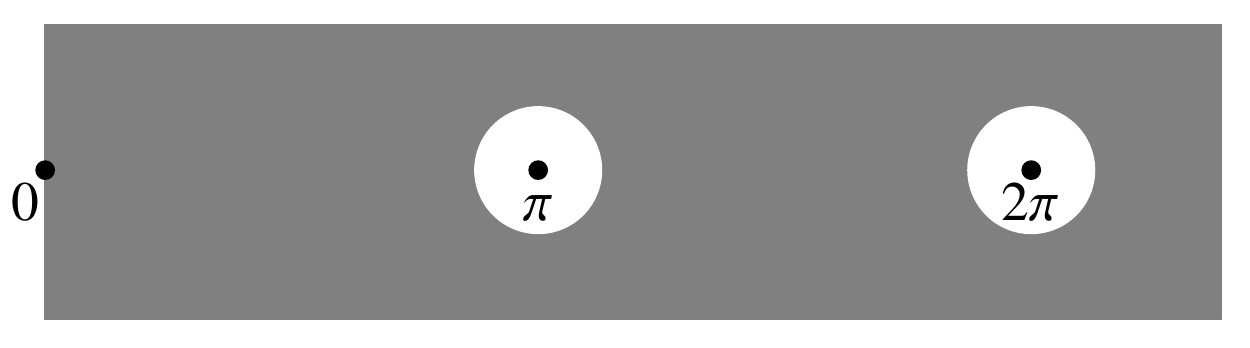}%
\caption{A typical domain of the form $S_{\varepsilon,A}.$}%
\label{sepsilon.fig}%
\end{figure}
With our notation established, we are ready to state our first main result.

\begin{theorem}
[Isometry Theorem]\label{isomHn.thm}Fix $f\in L^{2}(H^{2n+1})$ and let
$F=(e^{t\Delta/2}f)_{\mathbb{C}}.$ Then for all sufficiently small $R,$ the
integral%
\begin{equation}
I(R)=\int_{x\in H^{2n+1}}\int_{\substack{Y\in T_{x}(H^{2n+1})\\\left\vert
Y\right\vert \leq R}}\left\vert F(\exp_{x}(iY/2))\right\vert ^{2}\nu
_{2t}(\left\vert Y\right\vert )\frac{\sin^{2n}(\left\vert Y\right\vert
)}{\left\vert Y\right\vert ^{2n}}~dY~dx\label{isomForm}%
\end{equation}
is defined and convergent. Furthermore, $I(R)$ extends to a meromorphic
function on $\mathbb{C}$ with poles only at nonzero integer multiples of
$\pi,$ and%
\begin{equation}
\left\Vert f\right\Vert _{L^{2}(H^{2n+1})}^{2}=\lim_{R\rightarrow+\infty
}I(R),\label{isomLim}%
\end{equation}
where the limit is taken staying a fixed distance from the poles.
\end{theorem}

The factor $\sin^{2n}(\left\vert Y\right\vert )/\left\vert Y\right\vert ^{2n}$
is just the hyperbolic version of the Jacobian factor $J(Y)$ in
(\ref{compactIsom}). In this rank-one case, the unwrapped heat kernel
$\nu_{2t}$ depends only on $\left\vert Y\right\vert .$ Since the identity
component of $SO(3,1)$ admits a complex structure, the symmetric space $H^{3}$
belongs to the complex case. When $n=1$, the preceding result is a special
case of \cite[Thm. 3]{complex2}, after making a change of variables by a
factor of $2$ in $Y$. In the $n=1$ case, however, the function $I(R)$ has no
poles. 

We emphasize that the initial definition of $I(R)$ in (\ref{isomForm}) does
not make sense for large $R,$ because $F(\exp_{x}(iY))$ is only defined when
$Y$ is in a certain bounded region $\Omega,$ since $F:=(e^{t\Delta
/2}f)_{\mathbb{C}}$ does not extend to the whole complexification of
$H^{2n+1}.$ Thus, the limit on the right-hand side of (\ref{isomLim}) refers
to the meromorphic extension of $I(R).$

Now, since $H^{2n+1}$ has rank one, the $G$-orbit through $\exp_{x_{0}}%
(iY_{0}/2)$ consists of all points of the form $\exp_{x}(iY/2),$ where
$\left\vert Y\right\vert =\left\vert Y_{0}\right\vert .$ Thus, we may think of
the integral on the right-hand side of (\ref{isomForm}) as an orbital integral
through a point with $\left\vert Y_{0}/2\right\vert =R,$ followed by
integration with respect to $R.$ (See the first several paragraphs of Section
\ref{isometry.sec}.) Since, as we have noted, $\nu_{t}(r)$ has a pole of order
$2n-1$ at each nonzero (integer) multiple of $\pi,$ we see that the density in
(\ref{isomForm}) has a zero of order 1 at each nonzero multiple of $\pi.$
Except when $n=1,$ however, this zero in the density is \textit{not}
sufficient to cancel out the singularities in the orbital integrals.

It may seem surprising that integration of the singular orbital integral of
$F$ does not produce branching behavior in the function $I(R).$ The absence of
branching is demonstrated by integrating by parts to convert $I(R)$ into a
truncated version of the isometry formula in (\ref{kosIsom}). Since the
integrand on the right-hand side of (\ref{kosIsom}) is nonsingular, the only
singularities in $I(R)$ will come from the boundary terms, which do not
involve integration.

We turn, next, to the development of an inversion formula. We begin with an
inversion formula in the style of \cite{KOS}, which involves the shift
operator and applies to an arbitrary symmetric space of the noncompact type.
We then specialize this formula to the case of an odd-dimensional hyperbolic
space and perform an integration by parts, yielding a formula that is similar
to (\ref{isomForm}).

Let $G$ be a noncompact semisimple Lie group, assumed to be connected and with
finite center. Let $K$ be a maximal compact subgroup of $G.$ The quotient
$G/K$ has a $G$-invariant metric making $G/K$ into a Riemannian symmetric
space of the noncompact type. For each point $x$ in $G/K,$ let $K_{x}$ denote
the stabilizer of $x$ in $G,$ so that $K_{x}$ is conjugate to $K$ in $G.$ Then
let $f^{(x)}$ denote the average of $f$ over the action of $K_{x}$; that is,%
\[
f^{(x)}(y)=\int_{K_{x}}f(g\cdot y)~dg,
\]
where $dg$ is the normalized Haar measure on $K_{x}.$ If $F=e^{t\Delta/2}f,$
then since the heat operator commutes with all isometries of $G/K,$ we can say
that $F^{(x)}$ may be computed \textit{either} as the average of $F$ over the
action of $K_{x}$ \textit{or} as the heat operator applied to $f^{(x)}.$ Since
$F^{(x)}$ is invariant under the action of $K_{x},$ the function%
\[
Y\mapsto F^{(x)}(\exp_{x}(Y))
\]
is determined by its values on the tangent space $\mathfrak{a}_{x}$ to a
maximal flat through $x.$

We let $H^{r}(G/K),$ where $r$ is a positive real number, denote the Sobolev
space of functions on $G/K$ \textquotedblleft having $r$ derivatives in
$L^{2}.$\textquotedblright\ More precisely, consider the positive operator
$I-\Delta,$ viewed as an unbounded self-adjoint operator on $L^{2}(G/K).$ We
then take $H^{r}(G/K)$ to be the domain of the operator $(I-\Delta)^{r/2},$
where $(I-\Delta)^{r/2}$ is defined by the functional calculus for
self-adjoint operators. In Section \ref{inversion.sec}, we will describe the
Sobolev space $H^{r}(G/K)$ more concretely in terms of the Helgason Fourier
transform for $G/K.$

\begin{theorem}
[Inversion formula for an arbitrary symmetric space]\label{inversionGen.thm}%
Fix $f\in L^{2}(G/K)$ and let $F=e^{t\Delta/2}f.$ There is a positive constant
$r$ (depending only on $G$) such that if $f$ belongs to the Sobolev space
$H^{r}(G/K)$, we have%
\[
f(x)=e^{t\left\vert \delta\right\vert ^{2}/2}\int_{\mathfrak{a}_{x}}%
D(F^{(x)})(\exp_{x}(iY))\frac{e^{-\left\vert Y\right\vert ^{2}/(2t)}}{(2\pi
t)^{k/2}}~dY,
\]
with absolute convergence of the integral. Here $D$ is the shift operator
defined in Section 3.2 of \cite{KOS}.
\end{theorem}

Upon specializing Theorem \ref{inversionGen.thm} to the case of $H^{2n+1},$ we
may integrate by parts, with appropriate attention to the boundary terms, to
obtain an inversion formula similar to the isometry formula in Theorem
\ref{isomHn.thm}. As always, there is a factor of two difference in the
scaling of the variables between the isometry formula and the inversion formula.

\begin{theorem}
[Inversion formula for $H^{2n+1}$]\label{inversionHn.thm}Fix $f\in
L^{2}(H^{2n+1})$ and let $F=e^{t\Delta/2}f.$ Define, for all sufficiently
small $R,$%
\[
J(R,x)=\int_{\substack{Y\in T_{x}(H^{2n+1})\\\left\vert Y\right\vert \leq
R}}F(\exp_{x}(iY))\nu_{t}(\left\vert Y\right\vert )\frac{\sin^{2n}(\left\vert
Y\right\vert )}{\left\vert Y\right\vert ^{2n}}~dY.
\]
Then for each $x\in H^{2n+1},$ the function $J(R,x)$ extends to a meromorphic
function on $\mathbb{C}$ with poles only at nonzero integer multiples of
$\pi.$ Furthermore, there is a positive constant $r$ (depending only on $n$)
such that if $f$ belongs to the Sobolev space $H^{r}(H^{2n+1})$, we have%
\[
f(x)=\lim_{R\rightarrow+\infty}J(R,x),
\]
where the limit is taken along a path in $\mathbb{C}$ that stays a fixed
distance away from the poles.
\end{theorem}

The $n=1$ case of the preceding theorem is a special case of Theorem 6 in
\cite{complex1}. In the $n=1$ case, the function $J(R,x)$ actually has no poles.

We now turn to the question of a surjectivity theorem. In the complex case, we
proved that if $F$ is a holomorphic function for which $I(R)$ is defined for
small $R$ and has an analytic continuation with finite limit at $+\infty,$
then $F=(e^{t\Delta/2}f)_{\mathbb{C}}$ for some $f\in L^{2}(G/K).$ For
technical reasons that will be discussed in Section \ref{surjectivity.sec}, it
seems difficult to prove such a result in any case other than the complex
case. For the case of $H^{2n+1},$ $n\geq2,$ we content ourselves with the
following weaker result. In the following result, we make use of the Helgason
Fourier transform, as described in Section 2.3 of \cite{KOS}. If $f$ is a
function on $H^{2n+1}=G/K,$ the Fourier transform of $f,$ denoted $\hat{f},$
is a function on $\mathfrak{a}^{\ast}\times B,$ where $\mathfrak{a}%
\cong\mathbb{R}$ is the tangent space to a maximal flat in $H^{2n+1}$ and
where $B=M\backslash K,$ with $M$ being the centralizer of $\mathfrak{a}$ in
$K.$

\begin{theorem}
\label{surjectivity.thm}Suppose \thinspace$F$ is a holomorphic function on a
tube for which $I(R)$ is defined for all sufficiently small $R.$ Then the
restriction of $F$ to $H^{2n+1}$ is square integrable. Let $F_{\varepsilon}$
denote the function whose Fourier transform is given by%
\[
\hat{F}_{\varepsilon}(\lambda,b)=\hat{F}(\lambda,b)\mathbf{1}_{\{\left\vert
\lambda\right\vert <1/\varepsilon\}}%
\]
and let $I(R;F_{\varepsilon})$ denote the quantity in (\ref{isomForm}) with
$F$ replaced by $F_{\varepsilon}.$ If%
\[
\lim_{\varepsilon\rightarrow0}\lim_{R\rightarrow\infty}I(R;F_{\varepsilon})
\]
exists and is finite, there exists $f\in L^{2}(H^{2n+1})$ such that
$F=(e^{t\Delta/2}f)_{\mathbb{C}}.$
\end{theorem}

\section{Remarks on the Proofs of the Main Results}

At a fundamental level, the isometry and inversion formulas are based on a
duality between the geometry of the base, $G/K=H^{2n+1},$ and the geometry of
the fibers in the local identification of $T(G/K)$ with $G_{\mathbb{C}%
}/K_{\mathbb{C}}.$ As we have noted \cite[Prop. 1.17]{LGS}, the analytic
continuation of metric from the base restricts to the negative of a Riemannian
metric on a neighborhood of the identity in the fibers. The fibers, with the
resulting metric, are locally isometric to the dual compact symmetric space
$S^{2n+1}.$ For holomorphic functions on a tube in $G_{\mathbb{C}%
}/K_{\mathbb{C}},$ we have the following key result:%
\begin{align}
&  (\text{spherical Laplacian in fibers})\nonumber\\
&  =-(\text{hyperbolic Laplacian in base}).\label{fundId}%
\end{align}
\noindent(Compare \cite[Prop. 1.19]{LGS}.) This fundamental identity is
ultimately responsible for all of our main results. It follows formally from
(\ref{fundId}), for example, that doing the \textit{forward} heat operator in
the fibers accomplishes the \textit{backward} heat operator for the base,
which is precisely the inversion formula. When the base is compact, the
appropriate version of the fundamental identity is the key to proving both the
isometry and inversion formulas \cite{Hall2,St}.

When the base is noncompact, a convenient way to exploit the fundamental
identity is to use spherical functions. In the case of the isometry formula,
for example, there is a Gutzmer-type formula due to Faraut \cite{Fa,Fa2},
which says that the orbital integral in (\ref{ofDef}) may be computed as%
\begin{equation}
\mathcal{O}_{\left\vert F\right\vert ^{2}}(ir)=\int_{\mathbb{R}}\left\Vert
\hat{f}(\lambda)\right\Vert ^{2}e^{-t(\left\vert \lambda\right\vert
^{2}+\left\vert \delta\right\vert ^{2})}\phi_{\lambda}(ir)~d\mu(\lambda
),\label{ofGutz}%
\end{equation}
where $\left\Vert \hat{f}(\lambda)\right\Vert $ is the $L^{2}$ norm of the
Fourier transform $\hat{f}(\lambda,b)$ of $f$ over the $b$ variable and where
$\phi_{\lambda}$ is the spherical function with parameter $\lambda
\in\mathbb{R}.$ Now, as discussed in Section \ref{isometry.sec}, an
appropriate integral of $\mathcal{O}_{\left\vert F\right\vert ^{2}}(ir)$ over
$r$ gives the quantity $I(R)$ in the statement of Theorem \ref{isomHn.thm}.
Thus, (\ref{ofGutz}) becomes
\begin{align}
I(R) &  =c_{n}\int_{0}^{R}\mathcal{O}_{\left\vert F\right\vert ^{2}}%
(ir)\nu_{2t}(r)\sin^{2n}r~dr\nonumber\\
&  =\int_{\mathbb{R}}\left\Vert \hat{f}(\lambda)\right\Vert ^{2}%
e^{-t(\left\vert \lambda\right\vert ^{2}+\left\vert \delta\right\vert ^{2}%
)}\left[  c_{n}\int_{0}^{R}\phi_{\lambda}(ir)\nu_{2t}(r)\sin^{2n}r~dr\right]
~d\mu(\lambda).\label{irGutz}%
\end{align}

Meanwhile, $\phi_{\lambda}(r)$ is the restriction to a maximal flat of a
radial-type eigenfunction for the Laplacian for $H^{2n+1}.$ Thus, by
(\ref{fundId}), $\phi_{\lambda}(ir)$ is the restriction to a maximal flat of a
radial-type eigenfunction for the Laplacian for $S^{2n+1}.$ The expression in
square brackets on the right-hand side of (\ref{irGutz}) is then a
polar-coordinates computation of the integral of this eigenfunction against
the unwrapped spherical heat kernel $\nu_{2t}$.

Now, we expect that the integral of an eigenfunction of the Laplacian against
the heat kernel $\nu_{2t}$ should give $e^{t\lambda}$ times the value of the
eigenfunction at the basepoint, which is 1 in the case of a spherical
function. Since $\phi_{\lambda}(r)$ has eigenvalue $-(\left\vert
\lambda\right\vert ^{2}+\left\vert \delta\right\vert ^{2})$ for the hyperbolic
Laplacian, $\phi_{\lambda}(ir)$ has eigenvalue $\left\vert \lambda\right\vert
^{2}+\left\vert \delta\right\vert ^{2}$ for the spherical Laplacian. We
expect, then, that%
\begin{equation}
\lim_{R\rightarrow+\infty}c_{n}\int_{0}^{R}\phi_{\lambda}(ir)\nu_{2t}%
(r)\sin^{2n}r~dr=e^{t(\left\vert \lambda\right\vert ^{2}+\left\vert
\delta\right\vert ^{2})}.\label{spherHeat}%
\end{equation}
Thus, formally, letting $R$ tend to infinity in (\ref{irGutz}) should give%
\begin{align*}
\lim_{R\rightarrow\infty}I(R) &  =\int_{\mathbb{R}}\left\Vert \hat{f}%
(\lambda)\right\Vert ^{2}e^{-t(\left\vert \lambda\right\vert ^{2}+\left\vert
\delta\right\vert ^{2})}e^{t(\left\vert \lambda\right\vert ^{2}+\left\vert
\delta\right\vert ^{2})}~d\mu(\lambda)\\
&  =\left\Vert f\right\Vert ^{2},
\end{align*}
which is our isometry formula. A similar formal analysis yields the inversion
formula. (Start with (\ref{FxExpand}) in Section \ref{inversion.sec} and apply
(\ref{spherHeat}) with $t$ replaced by $t/2.$)

Of course, the preceding analysis is only formal, since it takes no account of
the singularities involved. In the case of odd-dimensional hyperbolic spaces,
the analytically continued spherical function $\phi_{\lambda}(ir)$ has
singularities at nonzero integer multiples of $\pi.$ The claim is that
$\phi_{\lambda}(ir)$ is actually meromorphic and that the above analysis can
be made rigorous if we simply integrate along a contour that avoids the poles.

To establish this claim, it is convenient to make use of shift operators. The
unwrapped heat kernel $\nu_{2t}$ can be expressed as a certain shift operator
applied to a Gaussian. Repeated integrations by parts move the shift operator
off the Gaussian and onto the spherical function, where it changes
$\phi_{\lambda}$ into its Euclidean counterpart, $\cosh(\lambda r)$. (See
Section \ref{shift.sec}.) Thus, after integrating by parts a finite number of
times, we end up with an integral (of a Gaussian times $\cosh(\lambda r)$)
that has no singularities and whose value can be computed explicitly. The
boundary terms in the integration by parts are manifestly meromorphic and they
tend to zero as $R$ tends to infinity, thus leading to a rigorous version of
(\ref{spherHeat}).

To prove the isometry formula, it still remains to interchange the limit as
$R$ tends to infinity with the integral on the right-hand side of
(\ref{irGutz}). Justifying this interchange for a general square-integrable
function $f$ requires sharp estimates on the analytically continued spherical
functions and their derivatives, which we obtain in Section \ref{shift.sec}.

For any symmetric space of the noncompact type, an analog of the fundamental
identity (\ref{fundId}) holds \cite{LGS}. If one can construct a suitable
unwrapped heat kernel and prove an analog of (\ref{spherHeat}), one may hope
to prove isometry and inversion formulas along the lines of what we have done
here for $H^{2n+1}.$

\section{Shift Operators and Spherical Functions\label{shift.sec}}

In this section, we consider various sorts of shift operators, each of which
has an \textquotedblleft intertwining property\textquotedblright\ relating the
radial part of a non-Euclidean Laplacian (for $H^{2n+1}$ or $S^{2n+1}$) to the
Euclidean Laplacian for $\mathbb{R}^{1}.$ There are a total of four shift
operators, two that shift (in one direction or the other) between $H^{2n+1}$
and $\mathbb{R}^{1}$ and two that shift between $S^{2n+1}$ and $\mathbb{R}%
^{1}.$ We also describe how the shift operators act on spherical functions and
use the resulting formulas to derive estimates on the spherical functions and
their derivatives.

Our first two shift operators are defined as follows:
\begin{align}
D^{\ast} &  =\left(  -\frac{1}{2\pi}\frac{1}{\sinh r}\frac{d}{dr}\right)
^{n}\label{Dstar1}\\
\tilde{D}^{\ast} &  =\left(  -\frac{1}{2\pi}\frac{1}{\sin r}\frac{d}%
{dr}\right)  ^{n}.\label{Dstar2}%
\end{align}
The star in the notation indicates that these operators are the adjoints of
other shift operators, which we will introduce shortly. It is known (e.g.,
\cite[p. 316]{CV}) and not hard to verify by direct calculation that these
shift operators have the following intertwining properties:%
\begin{align}
\left(  \frac{d^{2}}{dr^{2}}+2n\frac{\cosh r}{\sinh r}\frac{d}{dr}\right)
D^{\ast} &  =D^{\ast}\left(  \frac{d^{2}}{dr^{2}}-n^{2}\right)
\label{intertwine1}\\
\left(  \frac{d^{2}}{dr^{2}}+2n\frac{\cos r}{\sin r}\frac{d}{dr}\right)
\tilde{D}^{\ast} &  =\tilde{D}^{\ast}\left(  \frac{d^{2}}{dr^{2}}%
+n^{2}\right)  .\label{intertwine2}%
\end{align}
The operators in parentheses on the left-hand sides are the radial parts of
the Laplacians for $H^{2n+1}$ and $S^{2n+1},$ respectively.

The operators $D^{\ast}$ and $\tilde{D}^{\ast}$ may be used to convert the
heat kernels for $\mathbb{R}$ and for $S^{1}$ into the heat kernels for
$H^{2n+1}$ and $S^{2n+1},$ respectively. For us, the term \textquotedblleft
heat kernel\textquotedblright\ will always refer to the fundamental solution
of the heat equation%
\[
\frac{\partial u}{\partial t}=\frac{1}{2}\Delta u,
\]
where $\Delta$ is the Laplacian, which we take to be a \textit{negative} operator.

\begin{proposition}
\label{heatKernels.prop}The heat kernel $\gamma_{t}$ on $H^{2n+1}$ may be
computed as%
\[
\gamma_{t}(r)=e^{-tn^{2}/2}D^{\ast}\left(  \frac{1}{\sqrt{2\pi t}}%
e^{-r^{2}/(2t)}\right)
\]
and the heat kernel $\rho_{t}$ on $S^{2n+1}$ may be computed as%
\[
\rho_{t}(r)=e^{tn^{2}/2}\tilde{D}^{\ast}\left(  \frac{1}{\sqrt{2\pi t}}%
\sum_{k=-\infty}^{\infty}e^{-(r-2\pi k)^{2}/(2t)}\right)  ,
\]
where $r$ denotes the geodesic distance from the basepoint.
\end{proposition}

By comparing the above expression for $\rho_{t}$ to the definition of the
unwrapped heat kernel $\nu_{t}$ (Definition \ref{nut.def}) and by noting that
$\tilde{D}^{\ast}$ commutes with translations by $2\pi,$ we see that the
$2\pi$-periodization of $\nu_{t}$ is simply $\rho_{t}.$ This observation
justifies the term \textquotedblleft unwrapped heat kernel\textquotedblright%
\ for $\nu_{t}.$

\begin{proof}
Although this result is known \cite[Sect. 8]{Ca}, we briefly outline the
proof. It follows from (\ref{intertwine1}) and (\ref{intertwine2}) that both
$\gamma_{t}$ and $\rho_{t}$ satisfy the heat equation. It remains only to show
that these are the fundamental solutions. Clearly, both $\gamma_{t}$ and
$\rho_{t}$ will decay away from the basepoint, and near the basepoint, the
shift operators behave much like
\[
\frac{1}{(2\pi)^{n}}\left(  -\frac{1}{r}\frac{d}{dr}\right)  ^{n},
\]
which by explicit computation, takes $(2\pi t)^{-1/2}e^{-r^{2}/(2t)}$ to
$(2\pi t)^{-(2n+1)/2}e^{-r^{2}/(2t)}.$ Thus, for small $t,$ both $\gamma_{t}$
and $\rho_{t}$ look like the heat kernel for $\mathbb{R}^{2n+1},$ which in
turn behaves like a $\delta$-function.
\end{proof}

The formulas for the integral of a radial function on $H^{2n+1}$ and
$S^{2n+1}$ take the form%
\begin{align*}
&  c_{n}\int_{0}^{\infty}f(r)\sinh^{2n}r~dr\\
&  c_{n}\int_{0}^{\pi}f(r)\sin^{2n}r~dr,
\end{align*}
respectively, where $r$ is the geodesic distance from the basepoint and where
\[
c_{n}:=\frac{2(2\pi)^{n}}{(2n-1)!!}%
\]
is the surface area of the unit sphere in $\mathbb{R}^{2n+1}$. We now regard
$D^{\ast}$ and $\tilde{D}^{\ast}$ as maps from an $L^{2}$ space with Lebesgue
measure to a $L^{2}$ space with the measures coming from polar coordinates. We
then compute the adjoints of these maps (modulo boundary terms), which we call
$D$ and $\tilde{D},$ respectively.

\begin{proposition}
\label{Dform.prop}Let
\begin{align*}
D  &  =\frac{1}{(2n-1)!!}\prod_{k=1}^{n}\left(  \sinh r\frac{d}{dr}%
+(2k-1)\cosh r\right) \\
\tilde{D}  &  =\frac{1}{(2n-1)!!}\prod_{k=1}^{n}\left(  \sin r\frac{d}%
{dr}+(2k-1)\cos r\right)  ,
\end{align*}
where the product is taken with smaller values of $k$ to the left and larger
values of $k$ to the right. Then for all sufficiently smooth even functions
$f$ and $g$ on $[-R,R],$ we have%
\begin{align}
2\int_{0}^{R}(Df)(r)g(r)~dr  &  =\mathrm{B.T.}+c_{n}\int_{0}^{R}f(r)(D^{\ast
}g)(r)\sinh^{2n}r~dr\label{adjoint1}\\
2\int_{0}^{R}(\tilde{D}f)(r)g(r)~dr  &  =\mathrm{B.T.}+c_{n}\int_{0}%
^{R}f(r)(\tilde{D}^{\ast}g)(r)\sin^{2n}r~dr, \label{adjoint2}%
\end{align}
where \textquotedblleft$\mathrm{B.T.}$\textquotedblright\ indicates boundary
terms that involve the values of $f$ and $g$ and their derivatives at $R.$
Furthermore, $D$ and $\tilde{D}$ have the following intertwining properties:%
\begin{align}
D\left(  \frac{d^{2}}{dr^{2}}+2n\frac{\cosh r}{\sinh r}\frac{d}{dr}\right)
&  =\left(  \frac{d^{2}}{dr^{2}}-n^{2}\right)  D\label{intertwineD1}\\
\tilde{D}\left(  \frac{d^{2}}{dr^{2}}+2n\frac{\cos r}{\sin r}\frac{d}%
{dr}\right)   &  =\left(  \frac{d^{2}}{dr^{2}}+n^{2}\right)  \tilde{D}.
\label{intertwineD2}%
\end{align}

\end{proposition}

We emphasize that the boundary terms \textit{do not} involve the values of the
functions or their derivatives at 0, but only at $R.$

\begin{proof}
If $g$ is smooth and even on $[-R,R],$ then $g^{\prime}(r)/\sinh r$ has a
removable singularity at $r=0,$ and the resulting function is again smooth and
even on $[-R,R].$ Thus, there are no singularities in the computation of
$D^{\ast}g,$ despite the factors of $\sinh r$ in the denominator in the
definition of $D^{\ast},$ and similarly for $\tilde{D}^{\ast}g.$

We start on the right-hand side of, say, (\ref{adjoint1}) and successively
integrate by parts to move each factor in the definition of $D^{\ast}$ off of
$g$ and onto $f.$ To this end, we compute that%
\begin{align}
&  c_{k+1}\int_{0}^{R}f(r)\left[  \left(  -\frac{1}{2\pi}\frac{1}{\sinh
r}\frac{d}{dr}\right)  g(r)\right]  \sinh^{2(k+1)}r~dr\nonumber\\
&  =\mathrm{B.T.}+\frac{c_{k}}{2k+1}\int_{0}^{R}\left[  \sinh r~f^{\prime
}(r)+(2k+1)\cosh r~f(r)\right]  g(r)\sinh^{2k}r~dr,\label{partsOnce}%
\end{align}
where the boundary term comes from evaluating
\[
-(c_{k+1}/(2\pi))\sinh^{2k+1}r~f(r)g(r)
\]
at $0$ and at $R$ and where $c_{k+1}/(2\pi)=c_{k}/(2k+1).$ Note that even if
$k=0,$ the boundary term at $0$ vanishes. Integrating by parts $n$ times then
yields the first claimed identity, and an entirely similar calculation
verifies the second identity. Finally, since the radial part of each Laplacian
is a symmetric operator on a dense subspace of its respective Hilbert space,
(\ref{intertwineD1}) and\ (\ref{intertwineD2}) follow from (\ref{intertwine1})
and (\ref{intertwine2}) by taking adjoints. (We may initially prove the
desired identities on $C_{c}^{\infty}((0,\infty))$ and then extend to general
smooth functions by using the local nature of all operators involved.)
\end{proof}

\begin{lemma}
\label{Dorigin.lem}For any sufficiently nice function $f,$ we have%
\[
(Df)(0)=f(0).
\]

\end{lemma}

\begin{proof}
Although it is easy enough to verify this claim directly from the formula for
$D,$ it is more illuminating to use (\ref{adjoint1}). If we apply
(\ref{adjoint1}) with $g$ equal to the Euclidean heat kernel at time $t,$ then
$D^{\ast}g$ is the heat kernel at time $t$ for $H^{2n+1}$ multiplied by the
constant $a_{n}(t):=e^{tn^{2}/2}.$ Then the integral on the right-hand side of
(\ref{adjoint1}) is the integral of the radial extension of $f$ against the
heat kernel on $H^{2n+1}$ multiplied by $a_{n}(t),$ while the left-hand side
is the integral of the even extension of $Df$ against the heat kernel on
$\mathbb{R}.$ Letting $t$ tend to zero gives the claimed result, as the
boundary terms will vanish in the limit.
\end{proof}

\begin{proposition}
\label{Dspher.prop}The operator $D$ in Proposition \ref{Dform.prop} maps the
spherical function $\phi_{\lambda}$ for $H^{2n+1}$ to the function
$\cos(\lambda r).$ The operator $\tilde{D}$ in Proposition \ref{Dform.prop}
maps the analytically continued spherical function $\phi_{\lambda}(ir)$ to the
function $\cosh(\lambda r).$
\end{proposition}

\begin{proof}
The spherical function $\phi_{\lambda}$ is an even eigenfunction for the
radial part of the hyperbolic Laplacian with eigenvalue $-\lambda^{2}-n^{2}.$
Thus, by (\ref{intertwineD1}), the function $D\phi_{\lambda}$ is an
eigenfunction for $d^{2}/dr^{2}$ with eigenvalue $-\lambda^{2}.$ Since, also,
$D\phi_{\lambda}$ is even, we conclude that $D\phi_{\lambda}(r)$ is a constant
multiple of $\cos(\lambda r).$ But since $\phi_{\lambda}$ is normalized to
equal 1 at the basepoint, Lemma \ref{Dorigin.lem} tells us that the constant
is 1. From the definition of $\tilde{D},$ we now see that $\tilde{D}$ maps the
analytically continued spherical function $\phi_{\lambda}(ir)$ to
$(D\phi_{\lambda})(ir).$ That is, we have$\tilde{D}(\phi_{\lambda}%
(ir))=\cosh(\lambda r),$ the Euclidean counterpart of $\phi_{\lambda}(ir).$
\end{proof}

We next show that the adjoint operator $\tilde{D}^{\ast}$ maps $\cosh(\lambda
r)$ to a constant (depending on $\lambda$) times $\phi_{\lambda}(ir).$

\begin{lemma}
\label{spherical.lem}The analytically continued spherical function
$\phi_{\lambda}(ir)$ may be obtained from its Euclidean counterpart by the
action of $\tilde{D}^{\ast}$ as follows:
\[
\tilde{D}^{\ast}(\cosh(\lambda r))=c_{\lambda,n}\phi_{\lambda}(ir),
\]
where%
\[
c_{\lambda,n}=\frac{1}{(-2\pi)^{n}}\prod_{k=0}^{n-1}\frac{\lambda^{2}+k^{2}%
}{2k+1}.
\]

\end{lemma}

\begin{proof}
We let $\phi_{\lambda,k}$ denote the spherical function for $H^{2k+1},$ so
that $\phi_{\lambda,0}(ir)=\cosh(\lambda r).$ We then claim that%
\begin{equation}
\frac{1}{\sin r}\frac{d}{dr}\phi_{\lambda,k}(ir)=d_{\lambda,k}\phi
_{\lambda,k+1}(ir),\label{shiftkk1}%
\end{equation}
where%
\[
d_{\lambda,k}=\frac{\lambda^{2}+k^{2}}{2k+1}.
\]
To see this, let%
\[
\Delta_{k}:=\frac{d^{2}}{dr^{2}}+2k\frac{\cos r}{\sin r}\frac{d}{dr}%
\]
denote the radial part of the Laplacian on $S^{2k+1}.$ Direct calculation
shows%
\[
\Delta_{k+1}\left(  \frac{1}{\sin r}\frac{d}{dr}\right)  -\left(  \frac
{1}{\sin r}\frac{d}{dr}\right)  \Delta_{k}=(2k+1)\left(  \frac{1}{\sin r}%
\frac{d}{dr}\right)
\]
as operators on smooth functions of $r.$ Since $\phi_{\lambda,k}(r)$ is an
eigenfunction of the radial part of the Laplacian on $H^{2k+1}$ with
eigenvalue $-\lambda^{2}-k^{2},$ the function $\phi_{\lambda,k}(ir)$ is an
eigenfunction of $\Delta_{k}$ with eigenvalue $\lambda^{2}+k^{2}.$

Consequently,%
\begin{align*}
\Delta_{k+1}\left(  \frac{1}{\sin r}\frac{d}{dr}\phi_{\lambda,k}(ir)\right)
&  =\left(  \frac{1}{\sin r}\frac{d}{dr}\right)  \Delta_{k}(\phi_{\lambda
,k}(ir))+(2k+1)\left(  \frac{1}{\sin r}\frac{d}{dr}\phi_{\lambda
,k}(ir)\right)  \\
&  =(\lambda^{2}+(k+1)^{2})\left(  \frac{1}{\sin r}\frac{d}{dr}\phi
_{\lambda,k}(ir)\right)  .
\end{align*}
Both sides of (\ref{shiftkk1}) therefore are even, real-analytic
eigenfunctions for the radial part of the Laplacian on $S^{2k+3}$ with the
same eigenvalue. Thus, the two functions $\frac{1}{\sin r}\frac{d}{dr}%
\phi_{\lambda,k}(ir)$ and $\phi_{\lambda,k+1}(ir)$ are equal up to a constant,
as claimed in (\ref{shiftkk1}). (It is easy to check that the eigenfunction
equation has exactly one series solution in even powers of $r.$) To evaluate
the constant, we let $r$ tend to zero on both sides of (\ref{shiftkk1}), which
gives, by L'Hospital's rule,
\[
\phi_{\lambda,k}^{\prime\prime}(0)=d_{\lambda,k}\phi_{\lambda,k+1}%
(0)=d_{\lambda,k}.
\]

Meanwhile, $\phi_{\lambda,k}(0)=1$ and $\phi_{\lambda,k}(ir)$ satisfies%
\begin{equation}
\frac{d^{2}\phi_{\lambda,k}(ir)}{dr^{2}}+2k\frac{\cos r}{\sin r}\frac
{d\phi_{\lambda,k}(ir)}{dr}=(\lambda^{2}+k^{2})\phi_{\lambda,k}(ir).
\label{phiEqn}%
\end{equation}
Letting $r$ tend to zero in (\ref{phiEqn}) and using L'Hospital's rule again
gives%
\[
(2k+1)\phi_{\lambda,k}^{\prime\prime}(0)=\left(  \lambda^{2}+k^{2}\right)  ,
\]
which means that%
\[
d_{\lambda,k}=\phi_{\lambda,k}^{\prime\prime}(0)=\frac{\lambda^{2}+k^{2}%
}{2k+1},
\]
as claimed. If we apply (\ref{shiftkk1}) $n$ times and recall the definition
of $\tilde{D}^{\ast}$ in (\ref{Dstar2}), we obtain the lemma.
\end{proof}

We now estimate the analytically continued spherical function and its
derivatives. In what follows, we do not require that $r$ be real, but allow it
to range over a region of the form $S_{\varepsilon,A}.$

\begin{lemma}
\label{sphericalEst.lem}For $r$ in the region $S_{\varepsilon,A}$ (Notation
\ref{Sepsilon.notation}), the analytically continued spherical functions
satisfy the following estimate:
\begin{equation}
\left\vert \left(  \frac{d}{dr}\right)  ^{l}\phi_{\lambda}(ir)\right\vert \leq
C_{n,l}\frac{1+\left\vert r\right\vert }{(1+\left\vert \lambda\right\vert
)^{n-l-1}}\frac{e^{\left\vert \lambda r\right\vert }-1}{\left\vert \lambda
r\right\vert }. \label{phiEst}%
\end{equation}
In particular, for $\left\vert \lambda\right\vert >1$ and $\left\vert
r\right\vert >1,$ we have%
\begin{equation}
\left\vert \left(  \frac{d}{dr}\right)  ^{l}\phi_{\lambda}(ir)\right\vert \leq
D_{n,l}\frac{e^{\left\vert \lambda r\right\vert }}{\left\vert \lambda
\right\vert ^{n-l}}, \label{phiEstBig}%
\end{equation}
and for $\left\vert \lambda\right\vert \leq1,$ we have%
\begin{equation}
\left\vert \left(  \frac{d}{dr}\right)  ^{l}\phi_{\lambda}(ir)\right\vert \leq
E_{n,l}(1+\left\vert r\right\vert )e^{r}. \label{phiEstSmall}%
\end{equation}
Here $C_{n,l},$ $D_{n,l},$ and $E_{n,l}$ are constants that depend only on
$n,$ $l,$ and $\varepsilon.$
\end{lemma}

\begin{proof}
As in the proof of Lemma \ref{spherical.lem}, we let $\phi_{\lambda,k}$ denote
the spherical function with parameter $\lambda$ for $H^{2k+1}.$ We note that,
by (\ref{shiftkk1}),
\begin{align*}
\phi_{\lambda,1}(ir) &  =\frac{1}{\lambda^{2}}\frac{1}{\sin r}\frac{d}%
{dr}\cosh(\lambda r)\\
&  =\frac{\sinh(\lambda r)}{\lambda r}\frac{r}{\sin r}.
\end{align*}
Now, an elementary power-series argument shows that for all $l\geq0,$ we have
the estimate%
\[
\left\vert \left(  \frac{d}{dx}\right)  ^{l}\frac{\sinh x}{x}\right\vert
\leq\frac{e^{\left\vert x\right\vert }-1}{\left\vert x\right\vert }%
\]
so that%
\begin{align}
\left\vert \left(  \frac{d}{dr}\right)  ^{l}\frac{\sinh\lambda r}{\lambda
r}\right\vert  &  \leq\left\vert \lambda\right\vert ^{l}\frac{e^{\left\vert
\lambda r\right\vert }-1}{\left\vert \lambda r\right\vert }\nonumber\\
&  \leq(1+\left\vert \lambda\right\vert )^{l}\frac{e^{\left\vert \lambda
r\right\vert }-1}{\left\vert \lambda r\right\vert }.\label{sinhxx}%
\end{align}

We can now see inductively that $\phi_{\lambda,k}$ will be a finite linear
combination of terms of the form%
\[
d_{\lambda,k}\left[  \left(  \frac{d}{dr}\right)  ^{l}\frac{\sinh(\lambda
r)}{\lambda r}\right]  r^{a}g(r)
\]
where%
\[
d_{\lambda,k}\leq\frac{C}{(1+\left\vert \lambda\right\vert )^{2k-2}},
\]
where $a$ is either $0$ or $1,$ and where $g(r)$ is a rational expression in
$\sin r$ and $\cos r,$ with only sine factors in the denominator. In the
region $S_{\varepsilon,A},$ the function $g(r)$ will be bounded. Thus, by
(\ref{sinhxx}), we obtain the desired estimate.
\end{proof}

We conclude this section with a simple estimate that will be used in both
Sections \ref{isometry.sec} and \ref{inversion.sec}.

\begin{lemma}
\label{gaussEst.lem}For any positive number $a$ and positive integer $m,$
there is a constant $C$ such that for all $\lambda>0$, we have%
\[
\sup_{R}\left\vert e^{\lambda R}R^{m}e^{-aR^{2}/2}\right\vert \leq
C(1+\lambda^{m})e^{\lambda^{2}/(2a)},
\]
where the supremum is taken over $R$ in a region of the form $S_{\varepsilon
,A}$ (Notation \ref{Sepsilon.notation}).
\end{lemma}

\begin{proof}
By completing the square and writing $R^{m}=(R-\lambda/a+\lambda/a)^{m},$ we
obtain%
\[
e^{\lambda R}R^{m}e^{-aR^{2}/2}=e^{\lambda^{2}/(2a)}e^{-a(R-\lambda/a)^{2}%
/2}\sum_{l=0}^{m}\binom{m}{l}(R-\lambda/a)^{l}(\lambda/a)^{m-l}.
\]
Letting $x=R-\lambda/a$ and noting that the function $z\mapsto z^{l}%
e^{-az^{2}/2}$ is bounded on $S_{\varepsilon,A}$, we easily obtain the desired bound.
\end{proof}

\section{The Isometry Formula\label{isometry.sec}}

In this section, we provide the proof of Theorem \ref{isomHn.thm}. In simple
terms, the theorem follows from the isometry result of
Kr\"{o}tz--\'{O}lafsson--Stanton \cite{KOS} by taking the adjoint of the
operator $\tilde{D},$ which is done by means of integration by parts. Some
effort, however, is required to show that the boundary terms in the
integration by parts can be neglected.

Using (\ref{adjoint2}) in Proposition \ref{Dform.prop}, the $H^{2n+1}$ case of
the isometry formula in (\ref{kosIsom}) becomes, after taking into account the
symmetry of the orbital integral,%
\begin{align}
\left\Vert f\right\Vert ^{2} &  =\lim_{R\rightarrow\infty}2\int_{0}^{R}%
\tilde{D}\left[  \mathcal{O}_{\left\vert F\right\vert ^{2}}(ir)\right]
w_{t}(r)~dr\nonumber\\
&  =\lim_{R\rightarrow\infty}\left(  \mathrm{B.T.}+c_{n}\int_{0}%
^{R}\mathcal{O}_{\left\vert F\right\vert ^{2}}(ir)\tilde{D}^{\ast}\left[
w_{t}(r)\right]  \sin^{2n}r~dr\right)  .\label{isomPf1}%
\end{align}
Here the orbital integral is computed using the group $G$ equal to the
identity component of $SO(2n+1,1).$ From the formulas for $w_{t}$ and
$\tilde{D}^{\ast},$ we can see that $\tilde{D}^{\ast}[w_{t}(r)]$ coincides
with $\nu_{2t}(r),$ the unwrapped heat kernel at time $2t.$

Recall that for small $R,$ the map $(x,Y)\mapsto\exp_{x}(iY)$ is a
diffeomorphism of the set
\[
\{(x,Y)|~\left\vert Y\right\vert <R\}\subset T(H^{2n+1})
\]
onto its image in the complexification of $H^{2n+1},$ and that the image of
this diffeomorphism is called a tube. Since $H^{2n+1}$ has rank one, each
$G$-orbit in a tube will correspond under the diffeomorphism to a sphere
bundle:%
\begin{equation}
\{(x,Y)|~\left\vert Y\right\vert =a\}\subset T(H^{2n+1}). \label{sBundle}%
\end{equation}
Now, in the definition (\ref{ofDef}) of the orbital integral, the Haar measure
$dg$ on $G$ should be normalized so that the pushforward of $dg$ to $G/K$
coincides with the volume measure on $G/K.$ (That is to say, when $Y=0,$ the
orbital integral in (\ref{ofDef}) should coincide with the $L^{2}$ norm of $F$
over $G/K.$) Meanwhile, each $G$-orbit carries a $G$-invariant volume form
that is unique up to a constant. In a local trivialization of the bundle
(\ref{sBundle}), a $G$-invariant volume form may be constructed as a product
of the volume form on the sphere and the volume form on the base. If the
volume of the sphere is normalized to 1, the resulting volume form will
coincide with the integral against the Haar measure $dg,$ where $dg$ is
normalized as described above.

Therefore, if we compute the quantity $I(R)$ in the statement of Theorem
\ref{isomHn.thm} using polar coordinates, the integral over the sphere will
simply be an orbital integral. Thus,%
\begin{equation}
I(R)=c_{n}\int_{0}^{R}\mathcal{O}_{\left\vert F\right\vert ^{2}}(ir)\tilde
{D}^{\ast}\left[  w_{t}(r)\right]  \sin^{2n}r~dr, \label{IRpolar}%
\end{equation}
where we recognize the right-hand side of (\ref{IRpolar}) as the second term
on the right-hand side of (\ref{isomPf1}). To prove Theorem \ref{isomHn.thm},
then, we need only show that the boundary term on the right-hand side of
(\ref{isomPf1}) tends to zero as $R$ tends to infinity in $S_{\varepsilon,A}.$

To analyze the boundary terms, we use the Gutzmer-type formula of Faraut,
which also plays a key role in the results of \cite{KOS}. According to
\cite{Fa}, the orbital integral may be computed as%
\begin{equation}
\mathcal{O}_{\left\vert F\right\vert ^{2}}(ir)=\int_{\mathbb{R}}\left\Vert
\hat{f}(\lambda)\right\Vert ^{2}e^{-t(\left\vert \lambda\right\vert
^{2}+\left\vert \delta\right\vert ^{2})}\phi_{\lambda}(ir)~d\mu(\lambda
),\label{gutzmer}%
\end{equation}
where $\phi_{\lambda}$ is the spherical function with parameter $\lambda.$
Here, $\hat{f}$ is the Fourier transform of $f,$ viewed as a function
$\mathfrak{a}^{\ast}\times B,$ and $\left\Vert \hat{f}(\lambda)\right\Vert $
is the $L^{2}$ norm of $\hat{f}$ over $B,$ with $\lambda\in\mathfrak{a}^{\ast
}$ fixed. As usual, $\delta$ denotes half the sum of the positive roots with
multiplicity, and we have used the action of the heat operator on the Fourier
transform:
\[
\hat{F}(\lambda,b)=\hat{f}(\lambda,b)e^{-t(\left\vert \lambda\right\vert
^{2}+\left\vert \delta\right\vert ^{2})/2}.
\]
Finally, $\mu$ denotes a certain measure on $\mathfrak{a}^{\ast}%
\cong\mathbb{R},$ which may be computed in terms of the $c$-function. (When
$r=0,$ (\ref{gutzmer}) just gives the Plancherel theorem, expressing the
$L^{2}$ norm of $F$ over $H^{2n+1}$ in terms of its Fourier transform.)

The rapidly decaying factor of $e^{-t(\left\vert \lambda\right\vert
^{2}+\left\vert \delta\right\vert ^{2})}$ in (\ref{gutzmer}) makes it easy to
justify interchanging the operator $\tilde{D}$ with the integral, giving
\[
\tilde{D}\left[  \mathcal{O}_{\left\vert F\right\vert ^{2}}(ir)\right]
=\int_{\mathbb{R}}\left\Vert \hat{f}(\lambda)\right\Vert ^{2}e^{-t(\left\vert
\lambda\right\vert ^{2}+\left\vert \delta\right\vert ^{2})}\tilde{D}%
[\phi_{\lambda}(ir)]~d\mu(\lambda).
\]
Thus,%
\begin{align}
&  2\int_{0}^{R}\tilde{D}\left[  \mathcal{O}_{\left\vert F\right\vert ^{2}%
}(ir)\right]  w_{t}(r)~dr\nonumber\\
&  =\int_{\mathbb{R}}\left\Vert \hat{f}(\lambda)\right\Vert ^{2}%
e^{-t(\left\vert \lambda\right\vert ^{2}+\left\vert \delta\right\vert ^{2}%
)}\left[  2\int_{0}^{R}\tilde{D}[\phi_{\lambda}(ir)]w_{t}(r)~dr\right]
~d\mu(\lambda).\label{isomPf2}%
\end{align}
On the other hand, using (\ref{gutzmer}) in the expression of $I(R)$ given in
(\ref{IRpolar}) yields%
\begin{equation}
I(R)=\int_{\mathbb{R}}\left\Vert \hat{f}(\lambda)\right\Vert ^{2}%
e^{-t(\left\vert \lambda\right\vert ^{2}+\left\vert \delta\right\vert ^{2}%
)}\left[  c_{n}\int_{0}^{R}\phi_{\lambda}(ir)\tilde{D}^{\ast}[w_{t}%
(r)]\sin^{2n}(r)~dr\right]  ~d\mu(\lambda).\label{isomPf3}%
\end{equation}

By (\ref{adjoint2}) in Proposition \ref{Dform.prop}, the right-hand side of
(\ref{isomPf2}) and the right-hand side of Equation (\ref{isomPf3}) differ
just by a sequence of integrations by parts in the inner integral. Meanwhile,
the quantity on the right-hand side of (\ref{isomPf2}) admits an entire
analytic continuation as a function of $R,$ since Proposition
\ref{Dspher.prop} confirms $\tilde{D}[\phi_{\lambda}(ir)]$ is just the
Euclidean spherical function $\cosh(\lambda r),$ which has no singularities.
We must now analyze the boundary terms that arise if, say, we begin with the
(\ref{isomPf3}) and integrate by parts repeatedly in the inner integral. We
need to show that each boundary term is a meromorphic function of $R$ and that
these terms tend to zero if $R$ tends to infinity avoiding the poles.

To analyze the boundary terms, it is convenient to introduce the operator%
\[
L=-\frac{1}{2\pi}\frac{1}{\sin r}\frac{d}{dr},\text{ so that }\tilde{D}^{\ast
}=L^{n}.
\]
\bigskip The boundary terms of the integral%
\begin{equation}
\int_{0}^{R}\phi_{\lambda}(ir)\tilde{D}^{\ast}[w_{t}(r)]\sin^{2n}%
(r)~dr,\label{start}%
\end{equation}
are determined by successive integration by parts that each remove a power of
$L$ from operation on $w_{t}$. Using induction and (\ref{partsOnce}) from the
proof of Proposition \ref{Dform.prop}, we find that each boundary term is a
linear combination of terms of the form%
\begin{equation}
\left.  \left(  \left(  \frac{d}{dr}\right)  ^{l}\phi_{\lambda}(ir)\right)
\right\vert _{r=R}\cdot\sin^{j}R\cdot\cos^{k}R\cdot(L^{m}w_{t}%
)(R),\label{start2}%
\end{equation}
where $l+m\leq n-1$ and where all the exponents $l,$ $j,$ $k,$ and $m$ are
non-negative integers. Note, however, that there are negative powers of $\sin
r$ contained in the computation of $L^{m}w_{t}.$ It is easy to see that the
boundary terms tend to zero as $R\rightarrow\infty$ for each fixed $\lambda$;
what requires some effort is to justify interchanging the limit as
$R\rightarrow\infty$ with the integral over $\lambda$ in (\ref{isomPf3}).

Since the only singularities in each boundary term are the negative powers of
$\sin r$ arising from $L^{m}w_{t},$ we see that the boundary terms are
meromorphic with poles only at integer multiples of $\pi.$ Furthermore, since
$w_{t}(r)$ is an even function of $r,$ it is straightforward to check by
induction that $L^{m}w_{t}$ is even and nonsingular at the origin. Thus, the
singularities are only at \textit{nonzero} integer multiples of $\pi.$

Now, another application of induction confirms $(L^{m}w_{t})(R)$ is a linear
combination of terms of the form%
\[
e^{-\frac{R^{2}}{4t}}\cdot\sin^{p}R\cdot\cos^{q}R\cdot R^{d},
\]
where $d\leq m.$ Here, the exponents $q$ and $d$ are positive but the the
exponent $p$ of $\sin R$ may now be negative. Within the region
$S_{\varepsilon,A},$ positive powers of $\cos R$ and arbitrary powers of $\sin
R$ are bounded. Thus, for $\left\vert \lambda\right\vert \geq1,$ we can apply
(\ref{phiEstBig}) of Lemma \ref{sphericalEst.lem} to show that each
contribution to each boundary term is bounded by%
\[
C\frac{e^{\left\vert \lambda r\right\vert }}{\left\vert \lambda\right\vert
^{l}}R^{d}e^{-R^{2}/(4t)},
\]
where $l+d\leq l+m\leq n-1$ and where $C$ is a constant, independent of
$\lambda$ and $R.$

If we apply Lemma \ref{gaussEst.lem} with $a=1/(2t),$ we find that, for
$\left\vert \lambda\right\vert >1$ and $R>1,$ each term of the form
(\ref{start2}) is bounded by
\[
Ce^{t\left\vert \lambda\right\vert ^{2}}/R.
\]
Meanwhile, for $\left\vert \lambda\right\vert \leq1,$ we may apply
(\ref{phiEstSmall}) of Lemma \ref{sphericalEst.lem} to bound (\ref{start2})
by
\[
CR^{s}e^{R}e^{-R^{2}/(4t)}%
\]
for some $s.$

Since $\left\Vert \hat{f}(\lambda)\right\Vert ^{2}$ is integrable as a
function of $\lambda,$ it now follows by dominated convergence that the
boundary terms relating (\ref{isomPf2}) and (\ref{isomPf3}) tend to zero as
$\operatorname{Re}R\rightarrow+\infty$ in the region $S_{\varepsilon,A}.$

\section{The Surjectivity Theorem\label{surjectivity.sec}}

We would like to establish a result showing that if $F$ is any holomorphic
function on a tube for which $I(R)$ has an analytic continuation with finite
limit at infinity, then $F$ must be in the image of the heat operator. To
approach such a result, we apply the Gutzmer formula of Faraut, which tells us
that
\begin{equation}
I(R)=\int_{\mathbb{R}}\left\Vert \hat{F}(\lambda)\right\Vert ^{2}c_{n}\int%
_{0}^{R}\phi_{\lambda}(ir)\tilde{D}^{\ast}[w_{t}(r)]\sin^{2n}(r)~dr~d\mu
(\lambda). \label{IRsurj}%
\end{equation}
\textit{If} we could interchange the limit as $R$ tends to infinity (away from
the poles) with the outer integral, we could apply the rigorous version of
(\ref{spherHeat}) to conclude that the limit of $I(R)$ coincides with%
\[
\int_{\mathbb{R}}\left\Vert \hat{F}(\lambda)\right\Vert ^{2}e^{t(\left\vert
\lambda\right\vert ^{2}+\left\vert \delta\right\vert ^{2})}~d\mu(\lambda).
\]
\textit{If} we had this result and the limit of $I(R)$ were finite, we would
conclude that $F=e^{t\Delta/2}f,$ where $f$ is the function whose Fourier
transform is defined by $\hat{f}(\lambda,b)=\hat{F}(\lambda,b)e^{-t(\left\vert
\lambda\right\vert ^{2}+\left\vert \delta\right\vert ^{2})}.$

In the complex case, the desired interchange of limit and integral can be
justified by showing that the analog of the inner integral in (\ref{IRsurj})
is positive and monotone as a function of $R.$ (See the inner integral on the
right-hand side of Eq. (37) in Theorem 6 of \cite{complex2}.) We thus obtain a
\textquotedblleft strong\textquotedblright\ surjectivity result in the complex
case \cite[Thm. 8]{complex2}. In the case at hand, however, numerical
calculations indicate that the quantity%
\[
c_{n}\int_{0}^{R}\phi_{\lambda}(ir)\tilde{D}^{\ast}[w_{t}(r)]\sin^{2n}(r)~dr
\]
is neither positive nor monotone, except when $n=1$ (i.e., except in the
$H^{3}$ case, when $G$ is complex).

We can still prove, however, a weaker form of the surjectivity theorem for
$H^{2n+1},$ as in Theorem \ref{surjectivity.thm}.

\begin{proof}
[Proof of Theorem \ref{surjectivity.thm}]According to results of \cite{Fa2},
if $I(R)$ exists and is finite, the restriction of $F$ to $H^{2n+1}$ is
square-integrable. It therefore makes sense to construct the function
$F_{\varepsilon}.$ The function $F_{\varepsilon}$ is of the form
$F_{\varepsilon}=e^{t\Delta/2}f_{\varepsilon},$ where $f_{\varepsilon}$ is the
$L^{2}$ function with Fourier transform given by
\[
\hat{f}_{\varepsilon}(\lambda,b)=\hat{F}_{\varepsilon}(\lambda
,b)e^{t(\left\vert \lambda\right\vert ^{2}+\left\vert \delta\right\vert
^{2})/2}.
\]
We may therefore apply the isometry formula to $F_{\varepsilon}$, giving%
\[
\lim_{R\rightarrow\infty}I(R;F_{\varepsilon})=\left\Vert f_{\varepsilon
}\right\Vert ^{2}=\int_{\left\vert \lambda\right\vert \leq1/\varepsilon
}\left\Vert \hat{F}(\lambda)\right\Vert ^{2}e^{t(\left\vert \lambda\right\vert
^{2}+\left\vert \delta\right\vert ^{2})}~d\mu(\lambda).
\]
By monotone convergence, we have%
\[
\lim_{\varepsilon\rightarrow0}\lim_{R\rightarrow\infty}I(R;F_{\varepsilon
})=\int_{\mathbb{R}}\left\Vert \hat{F}(\lambda)\right\Vert ^{2}e^{t(\left\vert
\lambda\right\vert ^{2}+\left\vert \delta\right\vert ^{2})}~d\mu(\lambda).
\]
If the above expression is finite, we have $F=e^{t\Delta/2}f,$ where $f$ is
the $L^{2}$ function with Fourier transform given by%
\[
\hat{f}(\lambda,b)=\hat{F}(\lambda,b)e^{t(\left\vert \lambda\right\vert
^{2}+\left\vert \delta\right\vert ^{2})/2},
\]
thus establishing the claimed surjectivity result.
\end{proof}

\section{The Inversion Formula\label{inversion.sec}}

In this section, we provide proofs of the general inversion formula (Theorem
\ref{inversionGen.thm}) and the inversion formula for odd-dimensional
hyperbolic spaces (Theorem \ref{inversionHn.thm}).

\begin{proof}
[Proof of Theorem \ref{inversionGen.thm}]Fix a point $x\in G/K$, let $A_{x}$
be a maximal flat through $x,$ and let $\mathfrak{a}_{x}$ be the tangent space
at $x$ to $A_{x}.$ Since $f^{(x)}$ is invariant under the action of $K_{x},$
we may expand the restriction of $f^{(x)}$ to $A_{x}$ in terms of the
spherical functions (relative to the basepoint $x$):%
\[
f^{(x)}(a)=\int_{\mathfrak{a}_{x}^{\ast}}\phi_{\lambda}(a)\widehat{f^{(x)}%
}(\lambda)~\frac{d\lambda}{\left\vert c(i\lambda)\right\vert ^{2}},\quad a\in
A_{x}.
\]
(Since, after identifying $\mathfrak{a}_{x}$ with $\mathfrak{a},$ the
spherical functions are independent of $x,$ we suppress their dependence on
$x$ in the notation.) Here $\widehat{f^{(x)}}$ denotes the spherical Fourier
transform of $f^{(x)}$---which is essentially just the Helgason Fourier
transform for $G/K$ restricted to $K$-invariant functions---and $c(\cdot)$ is
the Harish-Chandra $c$-function.

Now, the Fourier transform \textquotedblleft diagonalizes\textquotedblright%
\ the action of the Laplacian $\Delta$ for $G/K$; specifically,
\[
\widehat{\Delta f}(\lambda,b)=(\left\vert \delta\right\vert ^{2}+\left\vert
\lambda\right\vert ^{2})\hat{f}(\lambda,b).
\]
Thus, the Sobolev space $H^{r}(G/K)$ may be described as the space of
functions $f\in L^{2}(G/K)$ for which%
\[
\int_{\mathfrak{a}^{\ast}}\left\vert \hat{f}(\lambda,b)\right\vert ^{2}\left(
\left\vert \delta\right\vert ^{2}+\left\vert \lambda\right\vert ^{2}\right)
^{r}~db~\frac{d\lambda}{\left\vert c(i\lambda)\right\vert ^{2}}<\infty.
\]
Now, if $f$ is in $H^{r}(G/K)$, then $f^{(x)}$ is also in $H^{r}(G/K).$ (This
claim holds because averaging over the action of $K$ commutes with the
Laplacian. Alternatively, averaging over the action of $K$ has the effect of
averaging $\hat{f}$ over the action of $K$ on $B,$ which only reduces the
$L^{2}$ norm of $\hat{f}$ over $B.$) In that case, the function
\[
\lbrack(I-\Delta)^{r/2}f^{(x)}]^{\symbol{94}}(\lambda)=\left(  1+\left\vert
\delta\right\vert ^{2}+\left\vert \lambda\right\vert ^{2}\right)
^{r/2}~\widehat{f^{(x)}}(\lambda)
\]
belongs to $L^{2}(\mathfrak{a},\left\vert c(\lambda)\right\vert ^{-2}).$ But
as a consequence of the Gindikin--Karpelevi\v{c} formula (see, e.g., Section
IV.6 of \cite{He}), the function $\left\vert c(i\lambda)\right\vert ^{-2}$ has
at most polynomial growth at infinity. Thus, the Cauchy--Schwarz inequality
tells us that%
\[
\widehat{f^{(x)}}(\lambda)=\frac{1}{\left(  1+\left\vert \delta\right\vert
^{2}+\left\vert \lambda\right\vert ^{2}\right)  ^{r/2}}\cdot\left(
1+\left\vert \delta\right\vert ^{2}+\left\vert \lambda\right\vert ^{2}\right)
^{r/2}~\widehat{f^{(x)}}(\lambda)
\]
is in $L^{1}(\mathfrak{a}_{x},\left\vert c(i\lambda)\right\vert ^{-2}),$
assuming $r$ is large enough.

Meanwhile, since $F^{(x)}=e^{t\Delta/2}(f^{(x)}),$ we have%
\begin{equation}
F^{(x)}(a)=\int_{\mathfrak{a}}\phi_{\lambda}(a)e^{-t(\left\vert \lambda
\right\vert ^{2}+\left\vert \delta\right\vert ^{2})/2}~\widehat{f^{(x)}%
}(\lambda)~\frac{d\lambda}{\left\vert c(i\lambda)\right\vert ^{2}}.
\label{FxExpand}%
\end{equation}
Thus, by the defining property of the shift operator (see Eqns. (3.11) and
(3.12) in \cite{KOS}),
\begin{equation}
D(F^{(x)})(\exp_{x}(Y))=\int_{\mathfrak{a}}\psi_{\lambda}(\exp_{x}%
(Y))e^{-t(\left\vert \lambda\right\vert ^{2}+\left\vert \delta\right\vert
^{2})/2}~\widehat{f^{(x)}}(\lambda)~\frac{d\lambda}{\left\vert c(i\lambda
)\right\vert ^{2}} \label{DFx}%
\end{equation}
where $\psi_{\lambda}$ is the Euclidean counterpart of the spherical function:%
\begin{equation}
\psi_{\lambda}(\exp_{x}(Y))=\frac{1}{\left\vert W\right\vert }\sum_{w\in
W}e^{i\left\langle \lambda,Y\right\rangle }. \label{psiLambda}%
\end{equation}
Here $W$ is the Weyl group for the symmetric space $G/K.$ We normalize the
$\psi_{\lambda}$'s differently from \cite{KOS} by including a factor of
$\left\vert W\right\vert $ in the denominator in (\ref{psiLambda}). Our shift
operator therefore differs by a factor of $\left\vert W\right\vert $ from that
in \cite{KOS}.

Now, by direct calculation,%
\[
\int_{\mathfrak{a}}e^{-\left\langle \lambda,Y\right\rangle }\frac
{e^{-\left\vert Y\right\vert ^{2}/(2t)}}{(2\pi t)^{k/2}}~dY=e^{t\left\vert
\lambda\right\vert ^{2}/2}.
\]
Since $\widehat{f^{(x)}}(\lambda)$ is in $L^{1},$ we may expand $\psi
_{\lambda}$ as a linear combination of exponentials and apply Fubini's theorem
to each of the $\left\vert W\right\vert $ terms in (\ref{DFx}). The result is
that%
\begin{align*}
&  e^{t\left\vert \delta\right\vert ^{2}/2}\int_{\mathfrak{a}_{x}}%
D(F^{(x)})(\exp_{x}(iY))\frac{e^{-\left\vert Y\right\vert ^{2}/(2t)}}{(2\pi
t)^{k/2}}~dY\\
&  =e^{t\left\vert \delta\right\vert ^{2}/2}\int_{\mathfrak{a}}e^{t\left\vert
\lambda\right\vert ^{2}/2}e^{-t(\left\vert \lambda\right\vert ^{2}+\left\vert
\delta\right\vert ^{2})/2}~\widehat{f^{(x)}}(\lambda)~\frac{d\lambda
}{\left\vert c(\lambda)\right\vert ^{2}}\\
&  =f^{(x)}(x)\\
&  =f(x),
\end{align*}
with absolute convergence of the integral.
\end{proof}

We now turn to the proof of the inversion formula for $H^{2n+1}$.

\begin{proof}
[Proof of Theorem \ref{inversionHn.thm}]Let us now specialize Theorem
\ref{inversionGen.thm} to the case of an odd-dimensional hyperbolic space,
$H^{2n+1}.$ Since the density against which we are integrating in the
definition of $J(R,x)$ is $K$-invariant, the result is unchanged if we replace
$F$ by $F^{(x)}.$ Thus, using polar coordinates and the definition of $\nu
_{t},$ we have
\begin{align}
&  J(R,x)=c_{n}\int_{0}^{R}F^{(x)}(ir)\nu_{t}(r)\sin^{2n}r~dr\nonumber\\
&  =c_{n}e^{tn^{2}/2}\int_{0}^{R}F^{(x)}(ir)\tilde{D}^{\ast}\left(
\frac{e^{-r^{2}/(2t)}}{\sqrt{2\pi t}}\right)  \sin^{2n}r~dr. \label{JRpolar}%
\end{align}
Thus, by Proposition \ref{Dform.prop},%
\begin{equation}
J(R,x)=\mathrm{B.T.}+e^{tn^{2}/2}2\int_{0}^{R}\tilde{D}(F^{(x)}(ir))\frac
{e^{-r^{2}/(2t)}}{(2\pi t)^{k/2}}~dr. \label{JRest}%
\end{equation}
After converting $2\int_{0}^{R}$ into $\int_{-R}^{R}$ and applying Theorem
\ref{inversionGen.thm}, we see that the last term on the right-hand side of
(\ref{JRest}) tends to $f(x).$ It then remains only to show that the boundary
terms tend to zero as $R$ tends to infinity in the region $S_{\varepsilon,A}.$

We now turn to the analysis of the boundary terms. We start from the last
expression in (\ref{JRpolar}) use the spherical Fourier transform to express
$F^{(x)}(ir)$, as follows:%
\begin{equation}
F^{(x)}(ir)=\int_{\mathbb{R}}\widehat{f^{(x)}}(\lambda)e^{-t(\left\vert
\lambda\right\vert ^{2}+\left\vert \delta\right\vert ^{2})/2}\phi_{\lambda
}(ir)\frac{d\lambda}{\left\vert c(i\lambda)\right\vert ^{2}} \label{FxExpand2}%
\end{equation}
(compare (\ref{FxExpand}) in the setting of general symmetric spaces of the
noncompact type). It is then straightforward to justify reversing the order of
integration, provided that we stay away from the poles in the $r$ integral.
Thus, we obtain that
\begin{align}
&  J(R,x)\nonumber\\
&  =\int_{\mathbb{R}}\widehat{f^{(x)}}(\lambda)e^{-t(\left\vert \lambda
\right\vert ^{2}+\left\vert \delta\right\vert ^{2})/2}\left[  c_{n}\int%
_{0}^{R}\phi_{\lambda}(ir)\tilde{D}^{\ast}\left(  e^{tn^{2}/2}\frac
{e^{-r^{2}/(2t)}}{\sqrt{2\pi t}}\right)  \sin^{2n}r~dr\right]  ~\frac
{d\lambda}{\left\vert c(i\lambda)\right\vert ^{2}}. \label{JRlast}%
\end{align}

We now successively integrate by parts in the inner integral in (\ref{JRlast}%
), and we must show that the boundary terms can be neglected. Fortunately, the
inner integral in (\ref{JRlast}) is precisely the same as the inner integral
in (\ref{isomPf3}), except that $t$ there has been replaced by $t/2$ here.
Thus, by the proof of Theorem \ref{isomHn.thm}, each boundary term can be
bounded by a constant (independent of $R$ and $\lambda$) times $e^{t\left\vert
\lambda\right\vert ^{2}/2}.$ Since, by our assumptions on $f,$ the function
$\widehat{f^{(x)}}(\lambda)$ is in $L^{1},$ we can apply dominated convergence
to move the limit as $R$ tends to infinity inside the integral, at which
point, the boundary terms certainly vanish.
\end{proof}


\acknowledgments{Acknowledgments}

The first author's research was supported in part by NSF grant DMS-1301534. The authors thank the two anonymous referees for a careful reading of the manuscript, which has improved the paper substantially.


\authorcontributions{Author Contributions}

Both authors contributed to the research and writing of this apper. 


\conflictofinterests{Conflicts of Interest}

The authors declare no conflict of interest.

\bibliographystyle{mdpi}

\begin{thebibliography}{999} 


\bibitem{AG}D. N. Akhiezer and S. G. Gindikin, On Stein extensions of real
symmetric spaces. \textit{Math. Ann.} \textbf{286} (1990), 1--12.

\bibitem{Ashtekar}A. Ashtekar, J. Lewandowski, D. Marolf, J.
Mour\~{a}o, and T. Thiemann, Coherent state transforms for spaces of
connections. \textit{J. Funct. Anal.} \textbf{135} (1996), 519--551.

\bibitem{Ba}V. Bargmann, On a Hilbert space of analytic functions and an
associated integral transform. \textit{Comm. Pure Appl. Math.} \textbf{14}
(1961), 187--214.

\bibitem{Ca}R. Camporesi, Harmonic analysis and propagators on homogeneous
spaces. \textit{Phys. Rep.} \textbf{196} (1990), 1--134.

\bibitem{CV}O. Chalykh and A. P. Veselov, Integrability and Huygens'
principle on symmetric spaces. \textit{Comm. Math. Phys.} \textbf{178} (1996), 311--338.

\bibitem{Da}E. B. Davies, Heat kernels and spectral theory, Cambridge
University Press, 1989.

\bibitem{DH}B. K. Driver and B. C. Hall, Yang-Mills theory and the
Segal-Bargmann transform. \textit{Comm. Math. Phys.} \textbf{201} (1999), 249--290.

\bibitem{Fa}J. Faraut, Formule de Gutzmer pour la complexification d'un
espace riemannien sym\'{e}trique. [A Gutzmer formula for the complexification
of a Riemannian symmetric space]. In \textquotedblleft Harmonic analysis on
complex homogeneous domains and Lie groups (Rome, 2001)\textquotedblright.
Atti Accad. Naz. Lincei Cl. Sci. Fis. Mat. Natur. Rend. Lincei (9) Mat. Appl.
13 (2002), no. 3-4, 233--241.

\bibitem{Fa2}J. Faraut, Analysis on the crown of a Riemannian symmetric
space, in: Lie Groups and Symmetric Spaces, in: Amer. Math. Soc. Transl. Ser.
2, vol. 210, Amer. Math. Soc., Providence, RI, 2003, pp. 99--110.

\bibitem{Gross1}L. Gross, Uniqueness of ground states for Schr\"{o}dinger
operators over loop groups. \textit{J. Funct. Anal.} \textbf{112} (1993), 373--441.

\bibitem{GM}L. Gross and P. Malliavin, Hall's transform and the
Segal-Bargmann map. In \textquotedblleft It\^{o}'s stochastic calculus and
probability theory\textquotedblright\ (N. Ikeda, S. Watanabe, M. Fukushima and
H. Kunita, Eds.) 73--116, Springer, Tokyo, 1996.

\bibitem{Hall1}B. C. Hall, The Segal-Bargmann \textquotedblleft coherent
state\textquotedblright\ transform for compact Lie groups. \textit{J. Funct.
Anal.} \textbf{122} (1994), 103--151.

\bibitem{Hall2}B. C. Hall, The inverse Segal-Bargmann transform for
compact Lie groups. \textit{J. Funct. Anal.} \textbf{143} (1997), 98--116.

\bibitem{Bull}B. C. Hall, Harmonic analysis with respect to heat kernel
measure. \textit{Bull. Amer. Math. Soc.} (N.S.) \textbf{38} (2001), 43--78.

\bibitem{HallYM}B. C. Hall, Coherent states and the quantization of
(1+1)-dimensional Yang-Mills theory. \textit{Rev. Math. Phys.} \textbf{13}
(2001), 1281--1305.

\bibitem{Geoquant}B. C. Hall, Geometric quantization and the generalized
Segal-Bargmann transform for Lie groups of compact type. \textit{Comm. Math.
Phys}. \textbf{226} (2002), 233--268.

\bibitem{complex1}B. C. Hall and J. J. Mitchell, The Segal-Bargmann
transform for noncompact symmetric spaces of the complex type. \textit{J.
Funct. Anal.} \textbf{227} (2005), 338--371.

\bibitem{complex2}B. C. Hall and J. J. Mitchell, Isometry theorem for the
Segal-Bargmann transform on a noncompact symmetric space of the complex type.
\textit{J. Funct. Anal.} \textbf{254} (2008), 1575--1600.

\bibitem{compactq}B. C. Hall and J. J. Mitchell, The Segal-Bargmann
transform for compact quotients of symmetric spaces of the complex type.
\textit{Taiwanese J. Math.} \textbf{16} (2012), 13--45.

\bibitem{HS}B. C. Hall and A. N. Sengupta, The Segal-Bargmann transform
for path-groups. \textit{J. Funct. Anal.} \textbf{152} (1998), 220--254.

\bibitem{He}S. Helgason, Groups and geometric analysis. Integral geometry,
invariant differential operators, and spherical functions. Corrected reprint
of the 1984 original. Mathematical Surveys and Monographs, 83. American
Mathematical Society, Providence, RI, 2000.

\bibitem{KOS}B. Kr\"{o}tz, G. \'{O}lafsson, and R. J. Stanton, The
image of the heat kernel transform on Riemannian symmetric spaces of the
noncompact type. \textit{Int. Math. Res. Not.} (2005), 1307--1329.

\bibitem{KS1}B. Kr\"{o}tz and R. J. Stanton, Holomorphic extensions of
representations. I. Automorphic functions. \textit{Ann. of Math.} (2)
\textbf{159} (2004), 641--724.

\bibitem{KS2}B. Kr\"{o}tz and R. J. Stanton, Holomorphic extensions of
representations. II. Geometry and harmonic analysis. \textit{Geom. Funct.
Anal.} \textbf{15} (2005), 190--245.

\bibitem{LGS}E. Leichtnam, F. Golse, and M. Stenzel, Intrinsic microlocal
analysis and inversion formulae for the heat equation on compact real-analytic
Riemannian manifolds. \textit{Ann. Sci. \'{E}cole Norm. Sup.} (4) \textbf{29}
(1996), 669--736.

\bibitem{OS1}G. \'{O}lafsson and H. Schlichtkrull, The Segal-Bargmann
transform for the heat equation associated with root systems. \textit{Adv.
Math.} \textbf{208} (2007), 422--437.

\bibitem{OS2}G. \'{O}lafsson and H. Schlichtkrull, Representation
theory, Radon transform and the heat equation on a Riemannian symmetric space.
In \textquotedblleft Group representations, ergodic theory, and mathematical
physics: a tribute to George W. Mackey,\textquotedblright\ 315--344, Contemp.
Math., 449, Amer. Math. Soc., Providence, RI, 2008.

\bibitem{Segal1}I. Segal, Mathematical problems of relativistic physics,
Chap.\thinspace VI, In: Proceedings of the Summer Seminar, Boulder, Colorado,
1960, Vol. II. (M. Kac, Ed.), Lectures in Applied Mathematics, American Math.
Soc., Providence, Rhode Island, 1963.

\bibitem{Segal2}I. Segal, The complex wave representation of the free
Boson field, In: Topics in functional analysis: Essays dedicated to M.G. Krein
on the occasion of his 70th birthday (I. Gohberg and M. Kac, Eds), Advances in
Mathematics Supplementary Studies, Vol. 3, Academic Press, New York, 1978, pp. 321-343.

\bibitem{St}M. B. Stenzel, The Segal-Bargmann transform on a symmetric
space of compact type. \textit{J. Funct. Anal.} \textbf{165} (1999), 44--58.

\bibitem{TW}T. Thiemann and O. Winkler, Gauge field theory coherent states
(GCS). II. Peakedness properties. \textit{Classical Quantum Gravity}
\textbf{18} (2001), 2561--2636.

\bibitem{Wr}K. K. Wren, Constrained quantisation and $\theta$-angles. II.
\textit{Nuclear Phys. B} \textbf{521} (1998), 471--502.
\end{thebibliography}
\makeatletter
\renewcommand\@biblabel[1]{#1. }
\makeatother

\end{document}